%% file: arxiv.tex
\documentclass{llncs}[11pt]
\usepackage{makeidx}  
\usepackage{graphicx}
\usepackage{setspace,algorithm,algorithmic,amsfonts}
\usepackage{amssymb, amsmath}
\usepackage{color}
\usepackage{verbatim}

\spnewtheorem{defi}{Definition}{\bfseries}{\rmfamily}
\spnewtheorem{redrule}{Reduction Rule}{\bfseries}{\rmfamily}
\spnewtheorem{obs}{Observation}{\itshape}{\rmfamily}
\spnewtheorem{cl}{Claim}{\bfseries}{\rmfamily}
\spnewtheorem{coro}{Corollary}{\bfseries}{\rmfamily}
\spnewtheorem{function}{Function}{\bfseries}{\rmfamily}
\newcommand{\LINEIF}[2]{%
    \STATE\algorithmicif\ {#1}\ \algorithmicthen\ {#2} \algorithmicend\ \algorithmicif%
}

\begin{document}

\pagestyle{headings}  

\title{Arbitrary Overlap Constraints in Graph Packing Problems}

\author{Alejandro L\'opez-Ortiz and Jazm\'in Romero}
\institute{David R.\ Cheriton School of Computer Science, University of Waterloo, Canada.}
\maketitle

\begin{abstract}
\input{Abstract}

\end{abstract}

\section{Introduction}
\input{Introduction}

\section{Preliminaries}
\input{Preliminaries}

\section{Packing Problems with Well-Conditioned Overlap}\label{GenericSection}

We start by developing an FPT-algorithm for the $r$-Set Packing with $\alpha()$-Overlap problem. After that, we provide a solution for $\Pi$-Packing with $\alpha()$-Overlap by reducing it to the set version. Our FPT-algorithm assumes that the function $\alpha()$ is well-conditioned. 

\subsection{An FPT Algorithm for the $r$-Set Packing with $\alpha()$-Overlap}\label{SetSection}
\input{BSTGeneric}

\subsection{The $\Pi$-Packing with $\alpha()$-Overlap problem}\label{GraphSection}
\input{Communities}

\section{Well-Conditioned Overlap Constraints}\label{AlphaConditionSection}

In the next section, we provide several examples of functions that are well-conditioned. That is, they satisfy the conditions in Definition \ref{alphaCondition}. In the first section, we focus on functions concerning the $r$-Set Packing with $\alpha()$-Overlap problem that by our discussion on Section \ref{GraphSection} could be used to restrict the overlap for graph version as well. After that in Section \ref{alphaGraphs}, we provide functions that consider graph properties. 

\subsection{Restricting the Overlap Between Sets}\label{alphaSets}
\input{AlphasForSetPacking}

\subsection{Restricting the Overlap Between Subgraphs}\label{alphaGraphs}
\input{AlphasForGraphPacking}

\section{Predetermined Cluster Heads}\label{PCHSection}
\input{PredeterminedClusterHeads}

 
\section{Conclusion}
\input{Conclusions}

\bibliographystyle{splncs03}
\bibliography{biblio}

\newpage
\section{Appendix}
\input{Appendix}

\end{document}

%% file: Abstract.tex
In earlier versions of the community discovering problem, the overlap between communities was restricted by a simple count upper-bound \cite{Mishra07,Cui13,Gossen14,ACMTOCT:FLOR}. In this paper, we introduce the $\Pi$-Packing with $\alpha()$-Overlap problem to allow for more complex constraints in the overlap region than those previously studied. 
Let $\mathcal{V}^r$ be all possible subsets of vertices of $V(G)$ each of size at most $r$, and $\alpha: \mathcal{V}^r \times \mathcal{V}^r \to \{0,1\}$ be a function.
The $\Pi$-Packing with $\alpha()$-Overlap problem seeks at least $k$ induced subgraphs in a graph $G$ subject to:
(i) each subgraph has at most $r$ vertices and obeys a property $\Pi$,  
and (ii) for any pair $H_i,H_j$, with $i\neq j$, $\alpha(H_i, H_j) = 0$ (i.e., $H_i,H_j$ do not \emph{conflict}).
We also consider a variant that arises in clustering applications:
each subgraph of a solution must contain a set of vertices from a given collection of sets $\mathcal{C}$, and no pair of subgraphs may share vertices from the sets of $\mathcal{C}$. In addition, we propose similar formulations for packing hypergraphs. 
We give an $O(r^{rk} k^{(r+1)k} n^{cr})$ algorithm for our problems where $k$ is the parameter and $c$ and $r$ are constants, provided that:
i) $\Pi$ is computable in polynomial time in $n$ and
ii) the function $\alpha()$ satisfies specific conditions. Specifically, $\alpha()$ is hereditary, applicable only to overlapping subgraphs, and computable in polynomial time in $n$. 
Motivated by practical applications we give several examples of $\alpha()$ functions which meet those conditions.

%% file: Introduction.tex
Many complex systems arising in the real world can be represented by networks, e.g. social and biological networks.
In these networks, a node represents an entity, and an edge represents a relationship between two entities.
A \emph{community} arises in a network when two or more entities have common interests.
In this way, members of a community tend to share several properties. 
Extracting the communities in a network is known as the \emph{community discovering problem} \cite{Fortunato10}.

In practice  communities may overlap by sharing one or more of their members \cite{Fortunato10,Coscia11,Xie13}. 
In \cite{BstAlgorithm14,ACMTOCT:FLOR}, the $\mathcal{H}$-Packing with $t$-Overlap was proposed as an abstraction for the community discovering problem. 
The goal is to find $k$ subgraphs in a given graph $G$ (the network) where each subgraph (a community) should be isomorphic to a graph $H \in \mathcal{H}$  where $\mathcal{H}$ is a family of graphs (the community models).
Every pair of subgraphs in the solution should not overlap by more than $t$ vertices (shared members).

However, in some cases the type of overlap that is allowed may be more complex.
For example, it has been observed in \cite{Yang14} that overlapping regions are denser than the rest of the community. 
Also, in \cite{Gossen14} it is suggested that overlapping regions should contain nodes which have a relationship with all the communities they belong to. 
Moreover, in \cite{Moustafa09} only \emph{boundaries nodes} can happen in the overlapping regions.
Motivated by this, we generalize the $\mathcal{H}$-Packing with $t$-Overlap to restrict the pairwise overlap by a function $\alpha()$ rather than by an upper-bound $t$. 
We also consider other communities models besides a family $\mathcal{H}$. 
The scope of community definitions is vast, see \cite{Fortunato10}. 
Thus, we define the much more general problem of \emph{$\Pi$-Packing with $\alpha()$-Overlap}. 

\begin{center}
\fbox{
\parbox{11.3cm}{
\textbf{The $\Pi$-Packing with $\alpha()$-Overlap problem}
    
    \noindent \emph{Input}: A graph $G$ and a non-negative integer $k$.
		
		\noindent \emph{Parameter}: $k$

    \noindent \emph{Question}: Does $G$ contain a $(k,\alpha)$-$\Pi$-packing, i.e., a set of at least $k$ induced subgraphs $\mathcal{K}=\{H_1, \dots ,H_k\}$ subject to the following conditions: i. each $H_i$ has at most $r$ vertices and obeys the property $\Pi$, and ii. for any pair $H_i,H_j$, with $i\neq j$, $\alpha(H_i,H_j)=0$?
}}
\end{center}

We also propose a similar generalization for the problem of packing sets with pairwise overlap that we call the \emph{$r$-Set Packing with $\alpha()$-Overlap problem}. Let $\mathcal{U}^r$ be all possible subsets of elements of $\mathcal{U}$ each of size at most $r$, and $\alpha: {\mathcal{U}^r} \times{\mathcal{U}^r} \to \{0,1\}$ be a function.

\begin{center}
\fbox{
\parbox{11.3cm}{
\textbf{The $r$-Set Packing with $\alpha()$-Overlap problem}
    
    \noindent \emph{Input}: A collection $\mathcal{S}$ each of size at most $r$, drawn from a universe $\mathcal{U}$, and a non-negative integer $k$.

    \noindent \emph{Parameter}: $k$

    \noindent \emph{Question}:  Does $\mathcal{S}$ contain a \emph{$(k,\alpha())$-set packing}, i.e., at least $k$ sets $\mathcal{K}=\{S_1,\dots , S_k\}$ where for each pair $S_i,S_j$ ($i \neq j$) $\alpha(S_i,S_j)=0$?}}
\end{center}

Some of our generalized problems are NP-complete; this follows from the NP-complete $H$-Packing and $r$-Set Packing problems. 
Our goal is to achieve \emph{fixed-parameter} (or \emph{FPT}) algorithms which are algorithms that provide a solution in $f(k)\, n^{O(1)}$ running time, where $f$ is some arbitrary computable function depending only on the parameter $k$. 
In all our problems, $k$ (the size of the solution) is the parameter, $r$ is a fixed constant, and $n$ denotes the order of the graph or the number of elements in the universe (depending on the problem).

\textbf{Related Work.}
H. Fernau et al., \cite{ACMTOCT:FLOR} provide an $O(r^r\, k^{r-t-1})$ kernel for the $\mathcal{H}$-Packing and $r$-Set Packing with $t$-Overlap problems. In addition, an $O(r^{rk}\,k^{(r-t-1)\,k+2}\,n^r)$ algorithm for these problems can be found in \cite{BstAlgorithm14}. A $2\,(rk-r)$ kernel when $\mathcal{H}= \{K_r\}$ and $t=r-2$ is given in \cite{BstAlgorithm14}.

The $H$-Packing problem has an $O(k^{r-1})$ kernel, where $H$ is an arbitrary graph on $r$ vertices.
Kernelization algorithms when $H$ is a prescribed graph can be found in \cite{Fellows04a,Fernau09,JanMar2014,Prieto06}. The $r$-Set Packing problem has an $O(r^r k^{r-1})$ kernel \cite{Faisal10}.

The community discovering problem is studied with a variety of approaches in \cite{Mishra07,Gregory07,Moustafa09,Arora12,Coscia11,Celebi15}, and comprehensive surveys are \cite{Fortunato10,Xie13}. 

\textbf{Our Results.} 
In this work, we introduce the $r$-Set Packing and $\Pi$-Set Packing with $\alpha()$-Overlap problems as more universal versions for the problem of packing graphs and sets subject to overlap constraints modeled by a function $\alpha()$. 
Our generalizations capture a much broader range of potential real life applications. 

We show in Section \ref{GenericSection} that the $r$-Set Packing with $\alpha()$ Overlap problem is fixed-parameter tractable when $\alpha()$ meets specific requirements ($\alpha()$ is \emph{well-conditioned}, see Definition \ref{alphaCondition}).
Our FPT-algorithm generalizes our previous algorithm~\cite{BstAlgorithm14}. Previously, we considered only a specific type a conflict between a pair of sets: overlap larger than $t$. In our extended algorithm, we will consider the more general \emph{$\alpha$-conflicts}.
To solve the $\Pi$-Packing with $\alpha()$-Overlap problem, we reduce it to its set version. 
This allows us to achieve an algorithm with $O(r^{rk}\, k^{(r+1)\,k}\, n^{cr})$ running time, provided that $\alpha()$ is well-conditioned and $\Pi$ is verifiable in polynomial time.


In Section \ref{AlphaConditionSection}, we give specific examples of well-conditioned $\alpha()$ functions, some motivated by practical applications while others by theoretical considerations. 
Specifically, a well-conditioned $\alpha()$ can restrict (but it is not limited to): 
i) the size of the overlap,
ii) the weight in the overlap region, (assuming as input a weighted graph), 
iii) \emph{the pattern} in the overlap region, i.e. the induced subgraph in the overlap should be isomorphic to a graph in $\mathcal{F}$, where $\mathcal{F}$ is a graph class that is hereditary, 
iv) that all overlapping vertices must satisfy a specific property $\xi$, 
v) that the overlap region should have a specific density, and
finally, v) the maximum distance between any pair of vertices in the overlap.

Lastly, we study the PCH-$r$-Set Packing with $\alpha()$-Overlap problem in Section \ref{PCHSection}.
In this setting, every set in the solution must contain a specific set of elements from a given collection of sets $\mathcal{C}$. 
This problem remain fixed-parameter tractable if $|\mathcal{C}|=O(g(k))$ for some computable function $g$ dependent on $k$ and independent of $n$.

%% file: Preliminaries.tex
Let $\mathcal{U}=\{u_1,\dots,u_n\}$ be a universe of elements and $\mathcal{S}=\{S_1,\dots,S_m\}$ be a collection of sets, where $S_i \subset \mathcal{U}$. We will use the letters $u$, $s$, $S$ in combination with subindices to refer to elements in $\mathcal{U}$, sets of elements of $\mathcal{U}$, and members of $\mathcal{S}$, respectively. Notice that we will identify a subset of elements of $\mathcal{U}$ (that is not necessarily a member of $\mathcal{S}$) using a lower-case $s$ with a subindex, while we restrict the use of upper-case letters to identify members of $\mathcal{S}$.

For $\mathcal{S'} \subseteq \mathcal{S}$, $val(\mathcal{S'})$ denotes the union of all members of $\mathcal{S'}$. 
We say that a subset of elements $s$ is \emph{contained} in a set $S$, if $s \subseteq S$.
In addition, let $\mathcal{S}(s)$ be the collection of all sets in $\mathcal{S}$ that contain $s$. 
That is, $s \subseteq S$ for each $S \in \mathcal{S}(s)$ and $s \not\subset S'$ for each $S' \in (\mathcal{S} \backslash \mathcal{S}(s))$. 
For any two sets $S,S'\in \mathcal{S}$, $|S \cap S'|$ is the \emph{overlap size} while $\{S \cap S'\}$ is the \emph{overlap region}. 

\begin{defi}\label{alphaCondition}
Let $\mathcal{U}^r$ be all possible subsets of elements of $\mathcal{U}$ each of size at most $r$, and $\alpha: {\mathcal{U}^r} \times{\mathcal{U}^r} \to \{0,1\}$ be a function. A pair of sets $s_i,s_j\in {\mathcal{U}^r}$ \emph{$\alpha$-conflict} if $\alpha(s_i,s_j)=1$ else they \emph{do not $\alpha$-conflict}. If $\alpha()$ satisfies the following requirements, we say $\alpha()$ is \emph{well-conditioned}.

\begin{itemize}
     \item [i)] $\alpha()$ is \emph{hereditary}. Specifically, if $s_i$ and $s_j$ do not $\alpha$-conflict ($\alpha({s_i},{s_j})=0$), $\alpha({s'_i},{s'_j})=0$ for any pair of subsets ${s'_i} \subseteq s_i$ and ${s'_j} \subseteq s_j$. 
		
   	 \item [ii)] If $s_i$ and $s_j$ $\alpha$-conflict ($\alpha(s_i,s_j)=1$), $|s_i \cap s_j| \geq 1$. 
		Furthermore, for any pair of subsets $s'_i \subseteq s_i$ and $s'_j \subseteq s_j$ with $\alpha({s'_i},{s'_j})=0$
		$((s_i \cap s_j) \backslash (s'_i \cap s'_j)) \neq \emptyset$.
		The elements in $s_i \cap s_j$ are referred to as the \emph{conflicting elements}.

		 \item [iii)] $\alpha$ is computable in polynomial time in $n$.
\end{itemize}

\end{defi}

A \emph{maximal $\alpha()$-set packing} $\mathcal{M} \subseteq \mathcal{S}$ is a maximal collection of sets from $\mathcal{S}$ such that for each pair of sets $S_i,S_j \in \mathcal{M}$ ($i \neq j$) $\alpha(S_i,S_j)=0$, and for each $S \in \mathcal{S}\backslash\mathcal{M}$, $S$ $\alpha$-conflicts with some $S' \in \mathcal{M}$, i.e., $\alpha(S,S')=1$.
 
All graphs in this paper are undirected and simple, unless otherwise stated.  For a graph $G$, $V(G)$ and $E(G)$ denote its sets of vertices and edges, respectively. $|V(G)|$ is the order of the graph. For a set of vertices $S \subseteq V(G)$, $G[S]$ represents the subgraph induced by $S$ in $G$. The distance (shortest path) between two vertices $u$ and $v$ is denoted as $dist_{G}(u,v)$.
We use the letter $n$ to denote both $|\mathcal{U}|$ and $|V(G)|$.

%% file: BSTGeneric.tex
The next lemmas state important observations of a maximal $\alpha()$-set packing and are key components in the correctness of our algorithm.

\begin{lemma}\label{maximalisaksolution}
Let $\mathcal{M}$ be a maximal $\alpha()$-set packing. If $|\mathcal{M}|\geq k$, then $\mathcal{M}$ is a $(k,\alpha())$-set packing.
\end{lemma}

\begin{proof}
Assume otherwise that $\mathcal{M}$ is not a $(k,\alpha())$-set packing. This would be only possible if there is at least one pair of sets $S_i,S_j$ in $\mathcal{M}$ for which $\alpha(S_i,S_j)=1$ but in that case $\mathcal{M}$ would not be a  maximal $\alpha()$-set packing. 
\end{proof}

\begin{lemma}\label{intersectionLemma}
Given an instance $(\mathcal{U},\mathcal{S},k)$ of $r$-Set Packing with $\alpha()$-Overlap, where $\alpha()$ is well-conditioned, let $\mathcal{K}$ and $\mathcal{M}$ be a $(k,\alpha())$-set packing and a maximal $\alpha()$-set packing, respectively. 
For each $S^* \in \mathcal{K}$, $S^*$ shares at least one element with at least one $S \in \mathcal{M}$. 
\end{lemma}

\begin{proof}
If $S^* \in \mathcal{M}$, the lemma simply follows. 
Assume by contradiction that there is a set $S^* \in \mathcal{K}$ such that $S^* \notin \mathcal{M}$ and there is no set $S \in \mathcal{M}$ $\alpha$-conflicting with $S^*$. However,  we could add $S^*$ to $\mathcal{M}$, contradicting its maximality. 
Thus, there exists at least one $S \in \mathcal{M}$ $\alpha$-conflicting with $S^*$.
Since $\alpha()$ is well-conditioned, by Definition \ref{alphaCondition} (ii) $|S \cap S^*| \geq 1$. \qed
\end{proof}

Our Bounded Search Tree algorithm (abbreviated as \texttt{BST-$\alpha()$-algorithm}) for $r$-Set Packing with $\alpha()$-Overlap has three main components: \texttt{Initialization}, \texttt{Greedy}, and \texttt{Branching}. We start by computing a maximal $\alpha()$-set packing $\mathcal{M}$ of $\mathcal{S}$. If $|\mathcal{M}|\geq k$ then $\mathcal{M}$ is a $(k,\alpha())$-set packing and the \texttt{BST-$\alpha()$-algorithm} stops (Lemma \ref{maximalisaksolution}). Otherwise, we create a search tree $T$ where at each node $i$, there is a collection of sets $\mathbf{Q^i}=\{s^i_1,\dots,s^i_k\}$ with $s^i_j \subseteq S$ for some $S \in \mathcal{S}$. The goal is \emph{to complete} $\mathbf{Q^i}$ to a solution, if possible. That is, to find $k$ sets $\mathcal{K}=\{S_1 \dots S_k\}$ of $\mathcal{S}$, such that $s^i_j \subseteq S_j$ for $1 \leq j \leq k$ and $\mathcal{K}$ is a $(k,\alpha())$-set packing. 

The children of the root of $T$ are created according to a procedure called \texttt{Initialization}. After that for each node $i$ of $T$, a routine called \texttt{Greedy} will attempt to complete $\mathbf{Q^i}$ to $(k,\alpha())$-set packing. If \texttt{Greedy} succeeds then the \texttt{BST-$\alpha()$-algorithm} stops. Otherwise, the next step is to create children of the node $i$ using the procedure \texttt{Branching}. The \texttt{BST-$\alpha()$-algorithm} will repeat \texttt{Greedy} in these children. 
Eventually, the \texttt{BST-$\alpha()$-algorithm} either finds a solution at one of the leaves of the tree or determines that it is not possible to find one. 

We next explain the three main components of the \texttt{BST-$\alpha()$-algorithm} individually. Let us start with the \texttt{Initialization} routine. By Lemma \ref{intersectionLemma}, if there is a solution $\mathcal{K}=\{S^*_1,\dots,S^*_k\}$ each $S^*_j$ contains at least one element of $val(\mathcal{M})$. 
Notice that each element of $val(\mathcal{M})$ could be in at most $k$ sets of $\mathcal{K}$. 
Thus, we create a set $\mathcal{M}_k$ that contains $k$ copies of each element in $val(\mathcal{M})$. That is, per each element $u \in val(\mathcal{M})$ there are $k$ copies $u_1 \dots u_k$ in $\mathcal{M}_k$ and $|\mathcal{M}_k| = k |val(\mathcal{M})|$. The root will have a child $i$ for each possible combination of $k$ elements from $\mathcal{M}_k$. A set of $\mathbf{Q^i}$ is initialized with one element of that combination. For example, if the combination is $\{u_1,u_2,u_k,a_1,b_1\}$, $\mathbf{Q^i}=\{\{u_1\},\{u_2\},\{u_k\},\{a_1\},\{b_1\}\}$. After that, we remove the indices from the elements in $\mathbf{Q^i}$, e.g., $\mathbf{Q^i}=\{\{u\},\{u\},\{u\},\{a\},\{b\}\}$. 

At each node $i$, the \texttt{Greedy} routine returns a collection of sets $\mathbf{Q^{gr}}$. 
Initially, $\mathbf{Q^{gr}} = \emptyset$ and $j=1$. 
At iteration $j$, \texttt{Greedy} searches for a set $S$ that contains $s^i_j \in \mathbf{Q^i}$ (the $j$th set of $\mathbf{Q^i}$) subject to two conditions (**): (1) $S$ is not already in $\mathbf{Q^{gr}}$ and (2) $S$ does not $\alpha$-conflict with any set in $\mathbf{Q^{gr}}$ (i.e., $\alpha(S,S')=0$ for each $S' \in \mathbf{Q^{gr}}$). 
If such set $S$ exists, \texttt{Greedy} adds $S$ to $\mathbf{Q^{gr}}$, i.e., $\mathbf{Q^{gr}} = \mathbf{Q^{gr}} \cup S$ and  continues with iteration $j=j+1$. 
Otherwise, \texttt{Greedy} stops executing and returns $\mathbf{Q^{gr}}$.
If $|\mathbf{Q^{gr}}|=k$, then $\mathbf{Q^{gr}}$ is a $(k,\alpha())$-set packing and the \texttt{BST-$\alpha()$-algorithm} stops.
If $\mathbf{Q^i}$ cannot be completed into a solution (Lemma \ref{terminationLemma}), \texttt{Greedy} returns $\mathbf{Q^{gr}} = \infty$.
\texttt{Greedy} searches for the set $S$ in the collection $\mathcal{S}(s^i_j,\mathbf{Q^i}) \subseteq \mathcal{S}(s^i_j)$ which is obtained as follows: add a set $S' \in \mathcal{S}(s^i_j)$ to $\mathcal{S}(s^i_j,\mathbf{Q^i})$, if $S'$ does not $\alpha$-conflict with any set in $(\mathbf{Q^i} \backslash s^i_j)$ and $S'$ is distinct of each set in $(\mathbf{Q^i} \backslash s^i_j)$.

The \texttt{Branching} procedure executes every time that \texttt{Greedy} does not return a $(k,\alpha())$-set packing but $\mathbf{Q^i}$ could be completed into one. That is, $\mathbf{Q^{gr}} \neq \infty$ and  $|\mathbf{Q^{gr}}|<k$.
Let $j=|\mathbf{Q^{gr}}|+1$ and $s^i_j$ be the $j$th set in $\mathbf{Q^i}$. 
\texttt{Greedy} stopped at $j$ because each set $S \in \mathcal{S}(s^i_j,\mathbf{Q^i})$ either it was already contained in $\mathbf{Q^{gr}}$, or  it $\alpha$-conflicts with at least one set in $\mathbf{Q^{gr}}$ (see **). 
We will use the conflicting elements between $\mathcal{S}(s^i_j,\mathbf{Q^i})$ and $\mathbf{Q^{gr}}$ to create children of the node $i$.
Let $I^*$ be the set of those conflicting elements. 
\texttt{Branching} creates a child $l$ of the node $i$ for each element $u_l \in I^*$. The collection $\mathbf{Q^l}$ of child $l$ is the same as the collection $\mathbf{Q^i}$ of its parent $i$ with the update of the set $s^i_j$ as $s^i_j \cup u_l$, i.e., $\mathbf{Q^l}=\{s^i_1,\dots,s^i_{j-1}, s^i_j \cup u_l, s^i_{j+1},\dots, s^i_{k}\}$. 
The set $I^*$ is obtained as $I^* = I^* \cup ((S \backslash s^i_j) \cap S')$ for each pair $S \in \mathcal{S}(s^i_j,\mathbf{Q^i})$ and $S' \in \mathbf{Q^{gr}}$ that $\alpha$-conflict ($\alpha(S,S')=1$) or that $S=S'$. The pseudocode of all these routines is detailed in the Appendix. 

\subsubsection{Correctness.}

With the next series of lemmas we establish the correctness of the \texttt{BST-$\alpha()$-algorithm} for any well-conditioned function $\alpha()$.

A collection $\mathbf{Q^i}=\{s^i_1,\dots,s^i_j,\dots,s^i_k\}$ is a \emph{partial-solution} of a $(k,\alpha())$-set packing $\mathcal{K}=\{S^*_1,\dots,S^*_j,\dots,S^*_k\}$ if and only if $s^i_j \subseteq S^*_j$, for $1 \leq j \leq k$. The next lemma states the correctness of the \texttt{Initialization} routine and it follows because we created a node for each selection of $k$ elements from $\mathcal{M}_k$, i.e., $\binom{\mathcal{M}_k}{k}$. 

\begin{lemma}\label{InitializationLemma}
If there exists at least one $(k,\alpha())$-set packing of $\mathcal{S}$, at least one of the children of the root will have a partial-solution.
\end{lemma}

\begin{proof}
By Lemma \ref{intersectionLemma}, every set in $\mathcal{K}$ contains at least one element of $val(\mathcal{M})$. 
It is possible that the same element be in at most $k$ different sets of $\mathcal{K}$. Therefore, we replicated $k$ times each element in $val(\mathcal{M})$ collected in $\mathcal{M}_k$. 
Since we created a node for each selection of $k$ elements from $\mathcal{M}_k$, i.e., $\binom{\mathcal{M}_k}{k}$, the lemma follows. \qed
\end{proof}

The next lemma states that the \texttt{BST-$\alpha()$-algorithm} correctly stops attempting to propagate a collection $\mathbf{Q^i}$.
Due to the (i) property of a well-conditioned $\alpha()$, we can immediately discard a collection $\mathbf{Q^i}$, if it has a pair of sets that $\alpha$-conflicts. In addition, the collection $\mathcal{S}(s^i_j,\mathbf{Q^i})$ contains all sets from $ \mathcal{S}(s^i_j)$ that are not $\alpha$-conflicting with any set in $\mathbf{Q^i}$ (excluding $s^i_j$). So again, if $\mathbf{Q^i}$ is a partial-solution, due to the (i) property, $\mathcal{S}(s^i_j,\mathbf{Q^i})$ cannot be empty. 

\begin{lemma}\label{terminationLemma}
Assuming $\alpha()$ is well-conditioned, $\mathbf{Q^i}$ is not a partial solution either: i. if there is a pair of distinct sets in $\mathbf{Q^i}$ that $\alpha$-conflict, or ii. if for some $s^i_j \in \mathbf{Q^i}$ $\mathcal{S}(s^i_j,\mathbf{Q^i})=\emptyset$.
\end{lemma}

\begin{proof}
(i) Suppose otherwise that $\mathbf{Q^i}$ is a partial-solution, but $s^i_j,s^i_l$ $\alpha$-conflict. 
Since $\mathbf{Q^i}$ is a partial-solution, $s^i_j \subseteq S^*_j$ and $s^i_l \subseteq S^*_l$ where  $S^*_j$,$S^*_l$ $\in \mathcal{K}$ and $\mathcal{K}$ is a $(k,\alpha())$-set packing. 

The pair $S^*_j$,$S^*_l$ does not $\alpha$-conflict, otherwise, $\mathcal{K}$ would not be a solution.
However, $\alpha()$ is hereditary, $s^i_j \subseteq S^*_j$, and $s^i_l \subseteq S^*_l$, thus, $s^i_j$ and $s^i_l$ do not $\alpha$-conflict either.

(ii) To prove the second part of the lemma, we will prove the next stronger claim. 

\indent \begin{cl}\label{FeasibleSponsorsLemma}
If $\mathbf{Q^i}=\{s^i_1,\dots, s^i_j, \dots, s^i_k\}$ is a partial-solution then $S^*_j \in \mathcal{S}(s^i_j,\mathbf{Q^i})$ for each $1 \leq j \leq k$.
\end{cl}

\begin{proof}
Assume by contradiction that $\mathbf{Q^i}$ is a partial-solution but $S^*_j \notin \mathcal{S}(s^i_j,\mathbf{Q^i},\alpha)$ for some $j$.

If $\mathbf{Q^i}$ is a partial-solution, $s^i_j \subseteq S^*_j \in \mathcal{K}$ and $((\mathbf{Q^i} \backslash s^i_j) \cup S^*_j$) is a partial-solution as well. 
The set $S^*_j \in \mathcal{S}(s^i_j)$ and $\mathcal{S}(s^i_j,\mathbf{Q^i}) \subseteq \mathcal{S}(s^i_j)$ (see Algorithm \ref{FeasibleSponsorsAlgorithm} for the computation of $\mathcal{S}(s^i_j,\mathbf{Q^i})$). 
The only way that $S^*_j$ would not be in $\mathcal{S}(s^i_j,\mathbf{Q^i})$ is if there is at least one set $S$ in $(\mathbf{Q^i} \backslash s^i_j)$ that $\alpha$-conflicts with $S^*_j$ or if $S^*_j$ is equal to a set in $(\mathbf{Q^i} \backslash s^i_j)$ but then $((\mathbf{Q^i} \backslash s^i_j) \cup S^*_j)$ would not be a partial-solution a contradiction to (i). \qed
\end{proof} 
\qed
\end{proof}

\texttt{Branching} creates at least one child whose collection is a partial-solution, if the collection of the parent is a partial-solution as well. Recall that $I^*$ is computed when \texttt{Greedy} stopped its execution at some $j \leq k$, i.e., it could not add a set that contains $s^i_j$ to $\mathbf{Q^{gr}}$. If $\mathbf{Q^i}$ is a partial solution, $s^i_j \subset S^*_j$ and $S^*_j \in \mathcal{K}$. Given property (ii) for a well-conditioned $\alpha()$, $S^*_j$ must be intersecting in at least one element with at least on set in $\mathbf{Q^{gr}}$. Therefore, at least one element of $S^*_j$ will be in $I^*$.

\begin{lemma}\label{BranchingLemma}
If $\mathbf{Q^i}=\{s^i_1,\dots,s^i_j,\dots,s^i_k\}$ is a partial-solution then there exists at least one $u_l \in I^*$ such that $\mathbf{Q^i}=\{s^i_1,\dots, s^i_j \cup u_l, \dots,s^i_k\}$ is a partial-solution.
\end{lemma}

\begin{proof}
Assume to the contrary that $\mathbf{Q^i}=\{s^i_1,\dots,s^i_j,\dots,s^i_k\}$ is a partial-solution but that there exists no element $u_l \in I^*$ such that $\mathbf{Q^i}=\{s^i_1,\dots, s^i_j \cup u_l, \dots,s^i_k\}$ is a partial-solution. 
This can only be possible if $(S^*_j \backslash s^i_j)  \cap I^* = \emptyset$.

First, given that $\mathbf{Q^i}$ is a partial-solution $S^*_j \in \mathcal{S}(s^i_j,\mathbf{Q^i},\alpha)$ (Claim \ref{FeasibleSponsorsLemma}).

In addition, either $S^*_j$ is already in $\mathbf{Q^{gr}}$ (i.e, it is equal to some set $S' \in \mathbf{Q^{gr}}$) or $S^*_j$ must be $\alpha$-conflicting with at least one set $S' \in \mathbf{Q^{gr}}$; otherwise, $S^*_j$ would have been selected by \texttt{Greedy}. 
Any of these situations implies that $|S^*_j \cap S'| \geq 1$ (Definition \ref{alphaCondition} (ii)).
By the computation of $I^*$ (Algorithm \ref{BranchRoutine}), $S^*_j \cap S' \subseteq I^*$.

Now it remains to show, that at least one element of $S^*_j \cap S'$ is in $S^*_j \backslash s^i_j$. This will guarantee that the set $s^i_j$ will be increased by one element at the next level of the tree. This immediately follows if $S^*_j = S'$.

Thus, we will show it for the case that $S^*_j \neq S'$ but $S^*_j$ $\alpha$-conflicts with $S'$.
Suppose that $(S^*_j \cap S') \cap (S^*_j \backslash s^i_j) = \emptyset$ by contradiction.
Recall that $S'$ contains some set $s^i_h$ of $\mathbf{Q^i}$ (for some $h \leq j$).
Furthermore, $S' \in \mathcal{S}(s^i_h,\mathbf{Q^i},\alpha)$; otherwise $S'$ would not have been selected by \texttt{Greedy}.

If $S'$ is $\alpha$-conflicting with $S^*_j$ but $S'$ is not $\alpha$-conflicting with $s^i_j$ (otherwise $S'$ would not have been in $\mathcal{S}(s^i_h,\mathbf{Q^i},\alpha)$), then $(S' \cap (S^*_j \backslash s^i_j)) \neq \emptyset$ by property (ii) in Definition \ref{alphaCondition}. \qed
\end{proof}

\begin{theorem}
The \texttt{BST-$\alpha()$-algorithm} finds a $(k,\alpha())$-set packing of $\mathcal{S}$, if $\mathcal{S}$ has at least one and $\alpha()$ is well-conditioned.
\end{theorem}

\subsubsection{Running Time.}
The number of children of the root is given by  $\binom{|\mathcal{M}_k|}{k} \leq \binom{k(r(k-1))}{k} = O((rk^2)^k)$ and the height of the tree is at most $(r-1)\,k$.  The number of children of each node at level $h$ is equivalent to the size of $I^*$ at each level $h$. The number of elements in $val(\mathbf{Q^{gr}})$ is at most $r(k-1)$, thus, $|I^*| \leq r(k-1)$. Therefore, the size of the tree is given by: $\binom{k\,(r(k-1))}{k} \, \prod_{h=1}^{(r-1)k} r\,(k-1)$ which is $O(r^{rk}\, k^{(r+1)\, k})$. In addition, $\alpha()$  is computable in $O(n^c)$ for some constant $c$ (property (iii), Definition \ref{alphaCondition}).

\begin{theorem}\label{runningtime}
The $r$-Set Packing with $\alpha()$-Overlap problem can be solved in \\ $O(r^{rk}\, k^{(r+1)\,k}\, n^{cr})$ time, when $\alpha()$ is well-conditioned.
\end{theorem}

%% file: Communities.tex
The $\Pi$-Packing with $\alpha()$-Overlap problem generalizes the $\mathcal{H}$-Packing with $t$-Overlap problem \cite{BstAlgorithm14} by including other community definitions in addition to prescribed graphs and by allowing more complex overlap restrictions.

We will represent a community through a graph property $\Pi$.
Intuitively, if a subgraph $H$ of order at most $r$ has the property $\Pi$ (called a \emph{$\Pi$-subgraph}), $H$ is a community. 
To obtain an FPT algorithm, we require however that $\Pi$ be verifiable in polynomial time in $n$ where $n=|V(G)|$.  

Examples of properties $\Pi$ that could represent communities are the following.
Let $S$ be an induced subgraph of $G$ with at most $r$ vertices.
$S$ is a community, if it has a density of at least $t$ ($|E(S)| \geq t$) and the number of edges connecting $S$ to rest of the network is at most a specific value \cite{Mishra07}.
$S$ is a community, if every vertex in $S$ is adjacent to at least $|V(S)|-c$ vertices in $S$ (for some constant $c$).
Observe that with our property $\Pi$, we still can use a family of graphs $\mathcal{H}$ to represent a community as in the $\mathcal{H}$-Packing with $t$-Overlap problem. In that case, $\Pi$ would correspond to the condition that $S$ is a community if $S$ is isomorphic to a graph $H$ in $\mathcal{H}$. 

In the $\Pi$-Packing with $\alpha()$-Overlap problem, we regulate the pairwise overlap with a function $\alpha: \mathcal{V}^r \times \mathcal{V}^r \to \{0,1\}$ where $\mathcal{V}^r$ is the collection of all possible subsets of vertices of $V(G)$ each of size at most $r$.
We say that two subgraphs $H_i$ and $H_j$ \emph{$\alpha$-conflict} if $\alpha(H_i,H_j)=1$.
Abusing the terminology, we extend the definition of a well-conditioned $\alpha()$ (Definition \ref{alphaCondition}) to consider subsets of vertices as well. This implies that $\mathcal{U}^r = \mathcal{V}^r$, $s_i=V(H_i)$ and $s_j=V(H_j)$ in Definition \ref{alphaCondition}. 

To provide a solution for the $\Pi$-Packing with $\alpha()$-Overlap problem, we will basically follow the approach of reducing this problem to the set version, i.e., to the $r$-Set Packing with $\alpha()$-Overlap problem. 
To this end, we first compute the collection of all induced $\Pi$-subgraphs of $G$, and we collect them in $\Pi_G$.
This is done by naively testing all sets of at most $r$ vertices from $G$.
We highlight that we are not asking to compute the largest subgraph of $G$ that follows $\Pi$, but rather only verifying whether a specific induced subgraph of at most $r$ vertices satisfies $\Pi$ or not.
In this way, $|\Pi_G| = O(n^{r})$.

Next, we construct an instance of $r$-Set Packing with $\alpha()$-Overlap as follows. The universe $\mathcal{U}$ equals $V(G)$ and there is a set $S=V(H)$ in $\mathcal{S}$ for each $H \in \Pi_G$. Furthermore, we require that $\alpha()$ be well-conditioned.

\begin{lemma}
The collection $\mathcal{S}$ has a $(k,\alpha())$-set packing if and only if $G$ has a $(k,\alpha())$-$\Pi$-packing.
\end{lemma}

\begin{proof}
We build a $(k,\alpha())$-set packing $\mathcal{K}_S$ from a $(k,\alpha)$-$\Pi$-packing. 
For each $\Pi$-subgraph $H_i$ in $\mathcal{K}$, we add a set $S_i = V(H_i)$ to $\mathcal{K}_S$. 
By our construction, $S_i \in \mathcal{S}$. 
Every pair of sets $S_i,S_j$ $\mathcal{K}_S$ do not $\alpha$-conflict. 
This follows because every pair $H_i,H_j \in \mathcal{K}_S$ do not $\alpha$-conflict, i.e.,  and $\alpha(H_i,H_j) = 0$.

Given a $(k,\alpha())$-set packing $\mathcal{K}_S$, we build a $(k,\alpha)$-$\Pi$-packing $\mathcal{K}$ of $G$. 
For each set $S_i$ in $\mathcal{K}_S$, we add a $\Pi$-subgraph $H_i=G[S_i]$. 
By our construction, $H_i$ is a $\Pi$-subgraph of $G$. 
Any pair of $\Pi$-subgraphs in $\mathcal{K}$ do not $\alpha$-conflict; otherwise, there would be a pair of sets in $\mathcal{K}_S$ $\alpha$-conflicting. \qed
\end{proof}

Given that $\Pi$ and $\alpha()$ are verifiable in $O(n^c)$ time for some constant $c$, we can hence state:

\begin{theorem}
$\Pi$-Packing with $\alpha()$-Overlap can be solved in $O(r^{rk}\, k^{(r+1)k}\, n^{cr})$ time, when $\alpha()$ is well-conditioned and $\Pi$ is polynomial time verifiable.
\end{theorem}

%% file: AlphasForSetPacking.tex


\vspace{0.15cm}
\noindent\textbf{Weighted Overlap.}
Let us assume that each $u_l \in \mathcal{U}$ has associated a non-negative weight $w(u_l)$. 
We could restrict the overlap region by its weight. 
The function $\alpha$-Weight($s_i,s_j$) returns \emph{no-conflict} if $w(s_i \cap s_j)= (\sum_{u \in (s_i \cap s_j)} w(u)) \leq w_t$ where $w_t \geq 0$ is a constant, else returns \emph{$\alpha$-conflict}.

Notice that we could use $\alpha$-Weight($s_i,s_j$) to upper-bound the overlap size by a constant $t$. To this end, we make $w(u)=1$ for each $u \in \mathcal{U}$, and $w_t=t$.

\begin{lemma}\label{weightProperties}
The function $\alpha$-Weight is well-conditioned.
\end{lemma}

\begin{proof}
(i) $\alpha$-Weight is \emph{hereditary}.
For any pair of sets $s_i$,$s_j$ with  $w(s_i \cap s_j) \leq w_t$ ($\alpha(s_i,s_j)=0$), there is no pair of subsets $s'_i \subseteq s_i,s'_j\subseteq s_j$ with $w(s'_i \cap s'_j) >w_t$ ($\alpha(s_i,s_j)=1)$. 
For the sake of contradiction, suppose otherwise.
Notice that $(s'_i \cap s'_j) \subseteq (s_i \cap s_j)$.
Thus, if $w(s_i \cap s_j) \leq w_t$ but $w(s'_i \cap s'_j) >w_t$ then there must be some elements in $(s_i \cap s_j) \backslash (s'_i \cap s'_j)$ with negative weights, a contradiction. 

(ii) If $\alpha(s_i,s_j)=1$ then $w(s_i \cap s_j) > w_t$ and $(s_i \cap s_j) \neq \emptyset$. 
Let $s'_i \subseteq s_i$ and  $s'_j \subseteq s_j$ be any pair of subsets with $\alpha(s'_i,s'_j)=0$, (i.e., $w(s'_i \cap s'_j) \leq w_t$).
Since $w(s_i \cap s_j) > w_t$, $w(s_i \cap s_j) - w(s'_i \cap s'_j) > 0$.
Therefore, $((s_i \cap s_j) \backslash (s'_i \cap s'_j)) \neq \emptyset$.

(iii) Finally, we can determine in $O(r)$ time, if $w(s_i \cap s_j) > w_t$. \qed
\end{proof}

We could also restrict the overlap region by both its size and its weight. This combined restriction is a well-conditioned function as well.

\vspace{0.15cm}
\noindent\textbf{Measures Overlap.}
A \emph{measure} of a set $S$ is a function $\mu$ that satisfies (i) $\mu(S) \geq 0$, (ii) $\mu(S)=0$ if $S = \emptyset$, and 
(iii) for any collection of pairwise disjoint subsets $S_1, \dots, S_l$ of $S$, $\mu(\bigcup^{l}_{i=1} S_i) = \sum^l_{i=1} \mu(S_i)$.
The last property implies that for any $S' \subseteq S$, $\mu(S') \leq \mu(S)$. 
Let $\mu$ be a measure on each set $\{s_i \cap s_j\}$ that is computable in polynomial time. The function $\alpha$-Measure($s_i,s_j)$ returns \emph{no-conflict} if $\mu(s_i \cap s_j) \leq t$ (where $t \geq 0$ is a constant) otherwise returns \emph{$\alpha$-conflict}.

\begin{lemma}
The function $\alpha$-Measure is well-conditioned.
\end{lemma}

\begin{proof}
(i) $\alpha$-Measure is \emph{hereditary}.
Assume by contradiction that $\mu(s_i \cap s_j) \leq t$ but there is pair of subsets $s'_i \subseteq s_i,s'_j\subseteq s_j$ with $\mu(s'_i \cap s'_j) > t$. Let $S = (s'_i \cap s'_j)$. 
First, $S \neq \emptyset$, otherwise, $\mu(S)=0$ and since $t \geq 0$ there would be a contradiction. 
Second $S \subseteq (s_i \cap s_j)$, thus by the additive property of $\mu$, $\mu(S) \leq \mu(s_i \cap s_j)$. Since $\mu(s_i \cap s_j) \leq t$, the claim holds.

(ii) If $\mu(s_i \cap s_j) > t$ then $|s_i \cap s_j| \geq 1$. This follows because $\mu(\emptyset) = 0$ and $t \geq 0$. 
Let $s'_i \subseteq s_i$ and  $s'_j \subseteq s_j$ be a pair of subsets with $\alpha(s'_i,s'_j)=0$, (i.e., $\mu(s'_i \cap s'_j) \leq t$).
Note that at most one $s'_i =s_i$ or $s'_j =s_j$; otherwise $\alpha(s'_i,s'_j)=1$.
Since $\mu(s_i \cap s_j) > t$, $\mu(s_i \cap s_j) - \mu(s'_i \cap s'_j) > 0$.
In this way, $((s_i \cap s_j) \backslash (s'_i \cap s'_j)) \neq \emptyset$.

(iii) The function $\mu$ is computed in polynomial time; thus, we can verify in constant time whether $\mu(s_i \cap s_j) > t$ or not. \qed
\end{proof}

\vspace{0.15cm}
\noindent\textbf{Metric Overlap.}
Let us assume that $\mathcal{U}$ is a metric space. That is, there is a \emph{metric} or \emph{a distance function} that defines a distance between each pair of elements $u,\,v$ of $\mathcal{U}$, subject to the following conditions: $dist_{\mathcal{U}}(u,v) \geq 0$, $dist_{\mathcal{U}}(u,v)=0$ if ($u=v$), $dist_{\mathcal{U}}(u,v)=dist_{\mathcal{U}}(v,u)$ and $dist_{\mathcal{U}}(u,w) \leq dist_{\mathcal{U}}(u,v)+dist_{\mathcal{U}}(v,w)$. 
For a constant $d_t>0$, we define the function $\alpha$-Metric($s_i,s_j$) which returns \emph{no-conflict} if $|s_i \cap s_j| \leq 1$ or $dist_{\mathcal{U}}(u,v) \leq d_t$ for each pair $u,v$ ($u \neq v$) in $s_i \cap s_j$ else returns \emph{$\alpha$-conflict}.

\begin{lemma}\label{metricLemma}
The function $\alpha$-Metric is well-conditioned.
\end{lemma}

\begin{proof}
(i) $\alpha$-Metric is hereditary.
For any pair of sets $s_i$,$s_j$ that do not $\alpha$-conflict (i.e., $\alpha(s_i,s_j)=0$), 
there is no pair of subsets $s'_i \subseteq s_i$, $s'_j\subseteq s_j$ with $\alpha(s'_i,s'_j)=1$. 
Assume the opposite by contradiction.
$|s'_i \cap s'_j| \geq 1$, otherwise $s'_i$ and $s'_j$ would not $\alpha$-conflict.
Observe that $(s'_i \cap s'_j) \subseteq (s_i \cap s_j)$. 
In addition, since $\mathcal{U}$ is a metric-space and $\alpha(s_i,s_j)=0$,
$dist_{\mathcal{U}}(u,v) \leq t$ for each pair $u,v$ ($u \neq v$) in $(s_i \cap s_j)$.
Given that we are using $dist_{\mathcal{U}}$ and not $dist_{s'_i \cap s'_j}$, there is no pair of elements in $(s'_i \cap s'_j)$ with $dist_{\mathcal{U}}(u,v) > t$.

(ii) Since $d_t>0$, for any pair $s_i,s_j$ with $\alpha(s_i,s_j)=1$, $|s_i \cap s_j| > 1$. 
Let $s'_i \subseteq s_i$ and $s'_j \subseteq s_j$ be a pair of subsets with $\alpha({s'_i},{s'_j})=0$.
Note that at most one $s'_i =s_i$ or $s'_j =s_j$; otherwise $\alpha(s'_i,s'_j)=1$.
$dist_{\mathcal{U}}(u,v) \leq t$ for each pair $u,v$ ($u \neq v$) in $(s'_i \cap s'_j)$.
Since $\alpha(s_i,s_j)=1$ but $\alpha(s'_i,s'_j)=0$ then it must exists at least one element $u$ in
$((s_i \cap s_j) \backslash (s'_i \cap s'_j))$ such that $dist_{\mathcal{U}}(u,v) > t$ for some $v$ in $(s_i \cap s_j)$.
In this way, $((s_i \cap s_j) \backslash (s'_i \cap s'_j)) \neq \emptyset$.

(iii) Assuming as input a metric space $\mathcal{U}$, $\alpha$-Metric is verified in $O(r^2)$ time.
\qed
\end{proof}

%% file: AlphasForGraphPacking.tex
\noindent\textbf{Prescribed Pattern.}
It has been observed in social networks that the overlap region is often more densely connected than the rest of the community \cite{Yang14}. 
Inspired by this, we will allow pairwise-overlap in a $(k,\alpha())$-$\Pi$-packing if the overlap region has a specific pattern, for example, it's a clique.
More precisely, we say that a pair of subgraphs $H_i$,$H_j$ do not $\alpha$-conflict if $G[V(H_i) \cap V(H_j)]$ is isomorphic to a graph $F$ in a class $\mathcal{F}$.
To define a well-conditioned $\alpha()$, $\mathcal{F}$ is a graph class that is hereditary (i.e., it is closed under taking induced subgraphs). 
To preserve our FPT results, any graph in $\mathcal{F}$ should be polynomial time verifiable. 
Examples of $\mathcal{F}$ are cliques, planar and chordal graphs. Indeed this applies to any
graph class that is closed under minors, since this is hereditary and by the Robertson-Seymour theorem the graph is polynomially testable by checking for the forbidden minors \cite{Robertson04}.
We define the function $\alpha$-Pattern($s_i,s_j$) that returns \emph{no-conflict} if $|s_i \cap s_j|=0$ or if $G[s_i \cap s_j]$ is isomorphic to a graph $F$ in $\mathcal{F}$; otherwise, it returns \emph{$\alpha$-conflict}.

\begin{lemma}\label{patternProperties}
The function $\alpha$-Pattern is well-conditioned.
\end{lemma}

\begin{proof}
(i) $\alpha$-Pattern is hereditary.
Assume by contradiction that there is pair $s_i$,$s_j$ with $\alpha(s_i,s_j)=0$ but there is a pair of subsets $s'_i \subseteq s_i$ and $s'_j \subseteq s_j$ with $\alpha(s'_i,s'_j)=1$.
If $\alpha(s'_i,s'_j)=1$ this implies that $G[s'_i \cap s'_j]$ is not isomorphic to a graph $F$ in $\mathcal{F}$.
Notice that $(s'_i \cap s'_j) \subseteq (s_i \cap s_j)$.
In addition, $(G[s_i \cap s_j])$ is isomorphic to a graph $F \in \mathcal{F}$ (otherwise, $\alpha(s_i,s_j)=1$).
Since $\mathcal{F}$ is a graph class that is hereditary, $G[s'_i \cap s'_j]$ is also isomorphic to $F$ and $s'_i$ and $s'_j$ do not $\alpha$-conflict.

(ii) It follows by definition of $\alpha$-Pattern($s_i,s_j$) that for any pair $s_i,s_j$ that $\alpha$-conflict (i.e., $\alpha(s_i,s_j)=1$) $|s_i \cap s_j| \geq 1$.
Let $s'_i \subseteq s_i$ and  $s'_j \subseteq s_j$ be a pair of non-empty subsets where $\alpha(s'_i,s'_j)=0$.
This implies that $G[s'_i \cap s'_j]$ is isomorphic to a graph $F \in \mathcal{F}$.
Recall that $(s'_i \cap s'_j) \subseteq (s_i \cap s_j)$.
Since $G[s_i \cap s_j]$ is not isomorphic to a graph in $\mathcal{F}$ but $G[s'_i \cap s'_j]$ is, and $\mathcal{F}$ is closed under taking induced subgraphs, $((s_i \cap s_j) \backslash (s'_i \cap s'_j)) \neq \emptyset$. 

(iii) Finally, $\alpha$-Pattern is computed in polynomial time as it is a constraint of the class $\mathcal{F}$.
\qed
\end{proof}

\vspace{0.15cm}
\noindent\textbf{Distance.}
In \cite{Moustafa09}, overlapping nodes occur only in the \emph{boundary} regions of overlapping communities in sensor networks. 
Motivated by this, we consider in the overlap region nodes that are ``closer'' to each other. In this way, two subgraphs $H_i$,$H_j$ do not $\alpha$-conflict if the distance in $G$ between any pair of vertices $u,v$ in $V(H_i) \cap V(H_j)$ is at most a constant $d_t > 0$, i.e., $dist_{G}(u,v)$ $\leq$ $d_t$. Recall that a subgraph $H_i$ is represented by a set $S_i=V(H_i)$ in $\mathcal{S}$. 
Since the graph distance is a metric on $V(G)$, we use the function $\alpha$-Metric defined previously (Lemma \ref{metricLemma}). 

\begin{lemma}
The function $\alpha$-Distance is an $\alpha$-Condition.
\end{lemma}

Note that we are using $dist_{G}(u,v) \leq d_t$ instead of $dist_{G[S_i \cap S_j]}(u,v) \leq d_t$. The second one is not an hereditary property and thus not well-conditioned. 

\vspace{0.15cm}
\noindent\textbf{Property.}
There are several vertex properties that are relevant to the analysis of real networks: vertex strength \cite{Chen10,Miritello13}, vertex weight \cite{Li13}, and disparity \cite{Miritello13}, among others. Hence, we suggest considering only overlapping nodes that present the same property $\xi$ (or properties). We assume however that the properties values for each vertex are given as part of the input. We define $\alpha$-Property($s_i,s_j$) which simply returns \emph{no-conflict} either if $|s_i \cap s_j| = 0$ or if each element $u$ in $\{s_i \cap s_j\}$ satisfies $\xi$. Otherwise, it returns \emph{$\alpha$-conflict}. 

\begin{lemma}
The function $\alpha$-Property is well-conditioned.
\end{lemma}

\begin{proof}
(i) $\alpha$ is \emph{hereditary}.
Assume by contradiction that there is pair $s_i$,$s_j$ with $\alpha(s_i,s_j)=0$ 
but there is a pair of subsets $\alpha(s'_i,s'_j)=1$ where $s'_i \subseteq s_i$ and $s'_j \subseteq s_j$. 
If $\alpha(s_i,s_j)=0$ every element in $s_i \cap s_j$ satisfies the property $\xi$. 
Since $(s'_i \cap s'_j) \subseteq (s_i \cap s_j)$, every element in $s'_i \cap s'_j$ satisfies $\xi$ as well.

(ii) By definition of $\alpha$-Property any pair of disjoint sets do not $\alpha$-conflict.
Let $s'_i \subseteq s_i$ and $s'_j \subseteq s_j$ be a pair of subsets with $\alpha({s'_i},{s'_j})=0$.
If $\alpha(s_i,s_j)=1$ but $\alpha({s'_i},{s'_j})=0$ there must exist at least one element in
$((s_i \cap s_j) \backslash (s'_i \cap s'_j))$ that does not follow $\xi$ and therefore
$((s_i \cap s_j) \backslash (s'_i \cap s'_j)) \neq \emptyset$.

(iii) The property $\xi$ for each element of $\mathcal{U}$ is given as part of the input. Thus, we can verify in constant time whether $s_i$, $s_j$ $\alpha$-conflict or not.
\qed
\end{proof}

\vspace{0.15cm}
\noindent \textbf{Dense Overlap.}
We design another $\alpha$ function to model the behavior that the overlap region is densely connected. To that end, we define $\alpha$-DenseOverlap($s_i,s_j$) that returns \emph{no-conflict} if $|s_i \cap s_j|=0$ or $|E(G[s_i \cap s_j])| \geq \frac{O(O-1)}{2} -c$, where $O=|s_i \cap s_j|$ and $c \geq 0$ is a constant; otherwise, it returns \emph{$\alpha$-conflict}. 

\begin{lemma}
The function $\alpha$-DenseOverlap is well-conditioned.
\end{lemma}

\begin{proof}
(i) $\alpha$-DenseOverlap is hereditary.
Assume by contradiction that there is pair of sets $s_i,s_j$ with $\alpha(s_i,s_j)=0$ but there is a pair of subsets $s'_i \subseteq s_i$ and $s'_j \subseteq s_j$ with $\alpha(s'_i,s'_j)=1$.
Notice that $(s'_i \cap s'_j) \neq \emptyset$; otherwise, $\alpha(s'_i,s'_j)=0$
Therefore, if $\alpha(s'_i,s'_j)=1$ then $|E(G[s'_i \cap s'_j])| < \frac{O'(O'-1)}{2} - c$, where $O'=|s'_i \cap s'_j|$.
However since $G[s'_i \cap s'_j]$ is an induced subgraph of $G[s_i \cap s_j]$, then $|E(G[s_i \cap s_j])| < \frac{O(O-1)}{2} -c$, a contradiction.

(ii) If $\alpha(s_i,s_j)=1$, $|s_i \cap s_j| \geq 1$. Furthermore, for any pair of subsets $s'_i \subseteq s_i$ and $s'_j \subseteq s_j$ with $\alpha(s'_i,s'_j)=0$ $(s_i \cap s_j) \backslash (s'_i \cap s'_j) \neq \emptyset$. 
Assume otherwise by contradiction.  
If $\alpha(s'_i,s'_j)=0$ then $|E(G[s'_i \cap s'_j])| \geq \frac{O'(O'-1)}{2} - c$, where $O'=|s'_i \cap s'_j|$.
Thus, if $(s_i \cap s_j) \backslash (s'_i \cap s'_j)$ would be the empty set, then $|E(G[s_i \cap s_j])| \geq \frac{O(O-1)}{2} - c$, where $O=|s_i \cap s_j|$, a contradiction to our assumption that $\alpha(s_i,s_j)=1$.

(iii) We can verify in polynomial time this condition.
\end{proof}

\vspace{0.15cm}
\noindent \textbf{Density.}
We could ask that the subgraph induced by the overlapping vertices has both at most $t$ vertices and $c$ edges. To that end, the function $\alpha$-Density returns \emph{no-conflict} if $|s_i \cap s_j|=0$ or ($|s_i \cap s_j| \leq {t}$ and $|E(G[s_i \cap s_j])| \leq {c}$), where $c \geq 0$, else returns $\alpha$-conflict.

\begin{lemma}
The function $\alpha$-Density is well-conditioned.
\end{lemma}

\begin{proof}
(i) $\alpha$-Density is \emph{hereditary}.
Assume by contradiction that there is pair $s_i$,$s_j$ with $\alpha(s_i,s_j)=0$,
but there is a pair of subsets $\alpha(s'_i,s'_j)=1$ where $s'_i \subseteq s_i$ and $s'_j \subseteq s_j$. 
Since $\alpha(s_i,s_j)=0$, both $(s_i \cap s_j) \leq t$ and $E(G[S_i \cap s_j])| \leq c$.
For any pair of sets $s_i$,$s_j$ with $\alpha(s_i,s_j)=0$, $|s_i \cap s_j| \leq t$. 
Thus, there cannot be a pair of subsets $s'_i \subseteq s_i,s'_j\subseteq s_j$ with $|s'_i \cap s'_j|>t$.
Given that $(s'_i \cap s'_j) \subseteq (s_i \cap s_j)$, $|E(G[s'_i \cap s'_j])| \leq |E({G[s_i \cap s_j]})| \leq c$.

(ii) For any pair $s_i,s_j$ that $\alpha$-conflict (i.e., $\alpha(s_i,s_j)=1$) $|s_i \cap s_j| \geq 1$.
In addition, for any pair of subsets $s'_i \subseteq s_i$ and  $s'_j \subseteq s_j$ with $\alpha(s'_i,s'_j)=0$ both $|s'_i \cap s'_j| \leq t$ and $E({G[s'_i \cap s'_j]})| \leq c$.
Since $\alpha(s_i,s_j)=1$ either $|s_i \cap s_j|>t$ or $|E({G[s_i \cap s_j]})| > {c}$.
In the first case,  since $\alpha(s'_i,s'_j)=0$, $|s'_i \cap s'_j| \leq t$. Therefore, $((s_i \cap s_j) \backslash (s'_i \cap s'_j)) \neq \emptyset$.
For the second case, $|E({G[s_i \cap s_j]})| - |E({G[s'_i \cap s'_j]})| > 0$.
Therefore, $|(s_i \cap s_j) \backslash (s'_i \cap s'_j)| \geq 1$.

(iii) In $O(r) $ time, we can verify if $|E({G[s_i \cap s_j]})| \leq {c}$.
\qed
\end{proof}

%% file: PredeterminedClusterHeads.tex
The problem of discovering communities in networks has been tackled with clustering algorithms as well \cite{Celebi15}. 
Many of these algorithms consider as part of the input a collection of sets of vertices $\mathcal{C}=\{C_1,\dots,C_l\}$ where each set $C_i \subset V(G)$ is called a \emph{cluster head}. 
The objective is to find a set of communities in $G$ where each community contains exactly one cluster head. In addition, communities should not share members of the cluster heads \cite{Tong14,Chen10,Li13,Cui13,Moustafa09}. 

Motivated by this, we introduce the  PCH-$r$-Set Packing with $\alpha()$-Overlap problem, where PCH stands for Predetermined Clusters Heads. The input of this problem is as before a universe $\mathcal{U}$, a collection $\mathcal{S}$, an integer $k$, but now it also has a collection of sets $\mathcal{C}=\{C_1, \dots C_l\}$ where $C_i \subset \mathcal{U}$. The goal there is to find  a $(k,\alpha())$-set packing (PCH), i.e., a set of at least $k$ sets $\mathcal{K}=\{S_1, \dots ,S_k\}$ subject to the following conditions: each $S_i$ contains at least one set of $\mathcal{C}$; for any pair $S_i$, $S_j$ with $i\neq j$, $(S_i \cap S_j) \cap val(\mathcal{C}) = \emptyset $, and $S_i,S_j$ do not $\alpha$-conflict. Recall that a $\Pi$-subgraph (or a community) is represented by a set in $\mathcal{S}$ (Section \ref{GraphSection}). Thus, this problem translates into a PCH variation for our $\Pi$-Packing problem as well. 

To solve the $r$-Set Packing with $\alpha$-Overlap problem (PCH), we need to do two modifications to the \texttt{BST-algorithm} described in Section \ref{GenericSection}.

First, we redefine the routine that creates the children of the root of the search tree, and we call it \texttt{Initialization (PCH)}.
By Lemma \ref{intersectionLemma}, a maximal solution $\mathcal{M}$ is used to determine the children of the root. 
In the (PCH)-variation, we no longer compute $\mathcal{M}$ but rather we use $\mathcal{C}$ to compute those children.
That is, the root will have a child $i$ for each possible combination of $\binom{\mathcal{C}}{k}$. 
Recall that a node $i$ has a collection $\mathbf{Q^i} = \{s^i_1,\dots,s^i_k\}$.
Each set of $\mathbf{Q^i}$ is initialized with set of that combination.

\begin{lemma}
If there exists at least one $(k,\alpha())$-set packing (PCH) of $\mathcal{S}$, at least one of the children of the root will have a partial-solution.
\end{lemma}

\begin{proof}
It follows by the explicit condition that each set in a $(k,\alpha())$-set packing (PCH) should contain at least one set from $\mathcal{C}$ and because the routine \texttt{Initialization (PCH)} tries all possible selections of size $k$ from $\mathcal{C}$ to create the children of the root.
\end{proof}

Second, we redefine the $\alpha$ function of the \texttt{BST-algorithm} as $\alpha$-PCH. 
This new function returns \emph{$\alpha$-conflict} if $((s_i \cap s_j) \cap val(\mathcal{C})) \neq \emptyset$; otherwise executes the original $\alpha()$ function and returns $\alpha$($s_i,s_j$).

\begin{lemma}
If the function $\alpha()$ is well conditioned, the function $\alpha$-PCH is also well-conditioned. 
\end{lemma}

\begin{proof}
(i) $\alpha$-PCH is hereditary. Assume that $\alpha$-PCH$(s_i,s_j)=0$, and there is pair of subsets $s'_i \subseteq s_i$ and $s'_j \subseteq s_j$ with $\alpha$-PCH$(s'_i,s'_j)=1$. Since $\alpha$ is well conditioned, this is only possible if 
$((s'_i \cap s'_j) \cap val(\mathcal{C})) \neq \emptyset$. However, $(s'_i \cap s'_j) \subseteq (s'_i \cap s'_j)$ and by our assumption $((s_i \cap s_j) \cap val(\mathcal{C})) = \emptyset$, a contradiction.

(ii) If $\alpha$-PCH$(s_i,s_j)=1$, $|s_i \cap s_j| \geq 1$. 
Assume otherwise by contradiction.  
Since $\alpha$ is well-conditioned, this is only possible if the extra condition in $\alpha$-PCH returns \emph{$\alpha$-conflict} when $\{s_i \cap s_j \}= \emptyset$.
However, in that case there can not be an intersection with the set $val(\mathcal{C})$, and $\alpha$-PCH$(s_i,s_j)=0$ instead.
It remains to show that for any pair of subsets $s'_i \subseteq s_i$ and $s'_j \subseteq s_j$ with $\alpha$-PCH$(s'_i,s'_j)=0$, $(s_i \cap s_j) \backslash (s'_i \cap s'_j) \neq \emptyset$.
Assume otherwise by contradiction, but in that case again there cannot be an intersection with $val(\mathcal{C})$ and $\alpha$-PCH$(s_i,s_j)=0$.

(iii) Since $\alpha$ is well-conditioned, and it takes $O(r)$ time to verify the extra condition in $\alpha$-PCH, $\alpha$-PCH is verified in polynomial time. \qed
\end{proof}

The above two modifications guarantee that the \texttt{BST Algorithm} will find a $(k,\alpha())$-Set Packing (PCH) if $\mathcal{S}$ has at least one. Given that each set in $\mathcal{S}$ has size at most $r$, we can immediately discard any set in $\mathcal{C}$ of size more than $r$. In this way, each set in $\mathcal{C}$ is be upper-bounded by a constant $c$,  $1 \leq c \leq r-1$.
To maintain our running time, the size of $\mathcal{C}$ should be $O(g(k))$, where $g$ is a computable function dependent only on $k$ and possibly $r$ but independent of $n$. Hence, we can state:

\begin{theorem}
If $\alpha()$ is well-conditioned, the PCH-$r$-Set Packing with $\alpha()$-Overlap problem is solved in $O((g(k))^{k} (rk)^{(r-1)k} n^{r})$ time, where $|\mathcal{C}| = g(k)$.
\end{theorem}

We could also omit the condition that clusters cannot share members of the cluster heads as in \cite{Dreier14}. In that case, we do not need to redefine the function $\alpha()$.

%% file: Conclusions.tex
We have proposed a more general framework for the problem of finding overlapping communities where the pairwise overlap meets a constraint function $\alpha()$. This framework captures much more realistic settings of the community discovering problem and can lead to interesting questions on its own.  We have also shown that our problems are fixed-parameter tractable when the overlap constraint $\alpha()$ is subject to a set of rather general conditions (Definition \ref{alphaCondition}). In addition, we have given several $\alpha()$ functions that meet those conditions. 

There are several interesting paths remaining to explore. It would be interesting to provide a fixed-parameter algorithm for our problems for functions other than those as in Definition \ref{alphaCondition}. For example, when the overlap is bounded by a percentage of the sizes of the communities or when the overlap size has a lower-bound instead of an upper-bound. In addition, a natural step would be to obtain kernelization algorithms for our problems.


%% file: Appendix.tex
\subsection{Pseudocode}

\begin{algorithm}
  \caption{\texttt{BST $\alpha()$-Algorithm}}
  \begin{algorithmic}[1]
    	\STATE{Compute a maximal $(\alpha())$-set packing $\mathcal{M}$}
			 
			\LINEIF{$|\mathcal{M}|\geq k$}{Return {$\mathcal{M}$}}
		
	   	\STATE{$T$=\texttt{Initialization}$(\mathcal{M})$}
	
       \FOR{each node $i$ of $T$}
			
           \STATE{Let $\mathbf{Q^i}$ be the collection of sets at node $i$}
					
			      \STATE{$\mathbf{Q^{gr}}=$\texttt{Greedy$(\mathbf{Q^i})$}}

						 \IF{$\mathbf{Q^{gr}}!=\infty$}
					  
						    \LINEIF{$|\mathbf{Q^{gr}}|=k$}{Return $\mathbf{Q^{gr}}$}

					      \STATE{\texttt{Branching}($T$,node $i$,$\mathbf{Q^i}$,$\mathbf{Q^{gr}}$)}
						\ENDIF
       \ENDFOR
  \end{algorithmic}\label{bstAlgorithm}
\end{algorithm}

\begin{algorithm}
  \caption{\texttt{Initialization}$(\mathcal{M})$}
  \begin{algorithmic}[1]
					 	\STATE{Replicate $k$ times each element $u \in val(\mathcal{M})$ and identify them as $u_1, \dots, u_k$.}
						\STATE{Let $\mathcal{M}_k$ be the enlarged set $val(\mathcal{M}$)}
						\STATE{$i=0$, $T=null$}
			      \WHILE{$i < |\binom{\mathcal{M}_k}{k}|$}
                 \STATE{Let $\mathbf{Q^i}=\{s^i_1,\dots,s^i_k\}$ be the $i$th combination of $\binom{\mathcal{M}_k}{k}$}
								 \STATE{\texttt{CreateNode($T$,root,node $i$,$\mathbf{Q^i}$)}}
								 \STATE{$i=i+1$}
						\ENDWHILE
						\STATE{Return $T$}
   \end{algorithmic}\label{IniRoutine}
\end{algorithm}

\begin{algorithm}\label{Greedy}
  \caption{\texttt{Greedy}($\mathbf{Q^i})$}
  \begin{algorithmic}[1]
  
	 \STATE{$\mathbf{Q^{gr}}=\infty$}
	 \STATE{//Check if $\mathbf{Q^i}$ could not be a partial solution}  
	 \IF{there is no pair $s^i_f$,$s^i_g$ in $\mathbf{Q^i}$ ($f \neq g$) with $\alpha(s^i_f,s^i_g)=1$}\label{Line_AlphaPartial_alpha}
	  \STATE{$\mathbf{Q^{gr}}=\emptyset$; $j=0$}   
    \REPEAT
     \STATE{Let $s^i_j$ be the $j$th set of $\mathbf{Q^i}$}

				\IF{$\mathcal{S}(s^i_j,\mathbf{Q^i},\alpha) = \emptyset$}
				      \STATE{$\mathbf{Q^{gr}}=\infty$}
				\ELSE

				\STATE{//Choose arbitrarily a set $S$ from $\mathcal{S}(s^i_j,\mathbf{Q^i},\alpha)$ such that}
				\STATE{//$S$ does not $\alpha$-conflict with any set in $\mathbf{Q^{gr}}$}
				\STATE{//and $S$ is not already in $\mathbf{Q^{gr}}$}
				
				\STATE{$f=0$}
				\WHILE{$f < |\mathcal{S}(s^i_j,\mathbf{Q^i},\alpha)|$}
				   \STATE{Let $S_f$ be the $f$-th set in $\mathcal{S}(s^i_j,\mathbf{Q^i},\alpha)$}
				   \STATE{$Conflicts = 0$}
				   \FOR{each $S' \in \mathbf{Q^{gr}}$}
					     \IF{$(\alpha(S_f,S')==1$) OR ($S_f==S'$)}
					        \STATE{$Conflicts = Conflicts + 1$}\label{Line_AlphaGr_alpha}  
								\ENDIF
					 \ENDFOR
					 \LINEIF{$Conflicts==0$}{$S=S_f$; $f =|\mathcal{S}(s^i_j,\mathbf{Q^i},\alpha)|+1$}
				\ENDWHILE
								
				\STATE{//Add the set $S$ to $\mathbf{Q^{gr}}$}
				\IF{such set $S$ does not exist}
				    \STATE{$j=k+1$}
				\ELSE
				 	  \STATE{$\mathbf{Q^{gr}}=\mathbf{Q^{gr}} \cup S$}
			  \ENDIF
				
	       \STATE{$j=j+1$}
		  \ENDIF
  \UNTIL{($j \geq k$) OR ($\mathbf{Q^{gr}=\infty}$)}
	\ENDIF
	
	\STATE{Return $\mathbf{Q^{gr}}$}
	
  \end{algorithmic}\label{GreedyRoutine}
\end{algorithm}

\begin{algorithm}
  \caption{\texttt{Branching}($T$,node $i$,$\mathbf{Q^i}$,$\mathbf{Q^{gr}}$)}
  \begin{algorithmic}[1]
     \STATE{Let $s^i_j$ be the first set of $\mathbf{Q^i}$ not completed by \textsc{Greedy}, i.e.,}
		 \STATE{$j=|\mathbf{Q^{gr}}|+1$ and $s^i_j=\mathbf{Q^i}[j]$}
		
     \STATE{$I^*=\emptyset$}
		 \FOR{each $S \in \mathcal{S}(s^i_j,\mathbf{Q^i},\alpha)$}\label{Line_ComputeIBegin_alpha}
	     \FOR{each $S' \in \mathbf{Q^{gr}}$}
			    \LINEIF{$\alpha(S,S')=1$ OR $(S==S')$}{$I^* = I^* \cup ((S \backslash s^i_j) \cap S')$}\label{Line_AlphaIstar_alpha}
				\ENDFOR
		 \ENDFOR\label{Line_ComputeIEnd_alpha}

     
     \STATE{$l=0$}
     
     \WHILE{$l \leq |I^*|$}
     
		     \STATE{Let $u_l$ be the $l$th element of $I^*$}
				
		     \STATE{$\mathbf{Q^l}=\{s^i_1,s^i_2,\dots,s^i_j \cup u_l, \dots, s^i_k\}$}\label{Line_BranchQL_alpha}
         
				 \STATE{\textsc{\textsc{CreateNode($T$,node $i$,node $l$,$\mathbf{Q^l}$)}}} 
         
         \STATE{$l=l+1$}
         
     \ENDWHILE  
  \end{algorithmic}\label{BranchRoutine}
\end{algorithm}

\begin{algorithm}
  \caption{\texttt{Compute} $\mathcal{S}(s^i_j,\mathbf{Q^i},\alpha$)}
  \begin{algorithmic}[1]
		\STATE{$l=0$, $\mathcal{S}(s^i_j,\mathbf{Q^i},\alpha) = \emptyset$}
			      \WHILE{$l < |\mathcal{S}(s^i_j)|$}
						     \STATE{Let $S_l$ be the $l$-th set in $\mathcal{S}(s^i_j)$}
								 \STATE{$f=0$, $conflicts= 0$}
								 \WHILE{$f  < |\mathbf{Q^i}|$}
								     \IF{$f \neq j$}
										      \IF{$\alpha(s^i_f,S_l)==1$ OR $(s^i_f==S_l)$}\label{Line_AlphaSponsors_alpha}
										         \STATE{$conflicts = conflicts + 1$}
													 \ENDIF
								     \ENDIF
										 \STATE{$f=f+1$}
								 \ENDWHILE
								 \LINEIF{$conflicts==0$}{$\mathcal{S}(s^i_j,\mathbf{Q^i},\alpha) = \mathcal{S}(s^i_j,\mathbf{Q^i},\alpha) \cup S_l$}
								 \STATE{$l=l+1$}
						\ENDWHILE
						\STATE{Return $\mathcal{S}(s^i_j,\mathbf{Q^i},\alpha)$}
   \end{algorithmic}\label{FeasibleSponsorsAlgorithm}
\end{algorithm}

\begin{algorithm}
  \caption{\texttt{Initialization (PCH)}$(\mathcal{C})$}
  \begin{algorithmic}[1]
						\STATE{$i=0$, $T=null$}
			      \WHILE{$i < |\binom{\mathcal{C}}{k}|$}
						     \STATE{Let $\{C^i_1,\dots,C^i_k\}$ be the $i$th combination of $\binom{\mathcal{C}}{k}$}
                 \STATE{Make $\mathbf{Q^i}=\{s^i_1,\dots,s^i_k\}$ equal to $\{C^i_1,\dots,C^i_k\}$, i.e, $s^i_j = C^i_j$}
								 \STATE{\texttt{CreateNode($T$,root,node $i$,$\mathbf{Q^i}$)}}
								 \STATE{$i=i+1$}
						\ENDWHILE
						\STATE{Return $T$}
   \end{algorithmic}\label{IniRoutinePCH}
\end{algorithm}